\documentclass[sigconf]{acmart}
\setcopyright{acmcopyright}
\copyrightyear{2025}
\acmYear{2025}

\acmConference[]{}{}{}
\acmBooktitle{}
\acmPrice{}
\acmISBN{}

\usepackage{amsthm,amsmath}
\usepackage{dsfont}
\usepackage{tikz}
\usepackage{todonotes}
\usepackage{hyperref}
\usepackage{nicefrac}
\usepackage{enumitem}

\newtheorem{theorem}{Theorem}

\newtheorem{example}[theorem]{Example}

\def\eatspace#1{#1}
\def\step#1#2{\par\kern1pt\hangindent#2em\hangafter=1\noindent\rlap{\small#1}\kern#2em\relax\eatspace}

\let\set\mathbb
\def\<#1>{\langle#1\rangle}
\def\im{\operatorname{im}}
\def\coker{\operatorname{coker}}
\def\kdim{\operatorname{kdim}}
\def\Res{\operatorname{Res}}

\def\nicefrac#1#2{#1/#2}

\def\ord{\operatorname{ord}}

\begin{document}
\fancyhead{}
\title{Bounds for D-Algebraic Closure Properties}
\thanks{%
  M.\ Kauers was supported by the Austrian FWF grants 10.55776/PAT8258123, 10.55776/I6130, and 10.55776/PAT9952223.
  R.\ Pages was supported by the Austrian FWF grants 10.55776/PAT8258123. 
}

\author{Manuel Kauers}
\affiliation{%
  \institution{Institute for Algebra}
  \institution{Johannes Kepler University}
  \state{}
  \city{4040 Linz}
  \country{Austria}
}

\email{manuel.kauers@jku.at}

\author{Rapha\"el Pag\`es}
\affiliation{%
  \institution{Institute for Algebra}
  \institution{Johannes Kepler University}
  \state{}
  \city{4040 Linz}
  \country{Austria}
}

\email{raphael.pages@jku.at}

\begin{abstract}
  We provide bounds on the size of polynomial differential
  equations obtained by executing closure properties for
  D-algebraic functions. While it is easy to obtain bounds
  on the order of these equations, it requires some more
  work to derive bounds on their degree. Here we give bounds
  that apply under some technical condition about the
  defining differential equations. 
\end{abstract}

\ccsdesc[500]{Computing methodologies~Algebraic algorithms}

\keywords{Computer algebra, differential equations, elimination theory}
\maketitle

\section{Introduction}

D-finite functions have been a prominent topic in computer algebra for
many years. They are defined as solutions of linear differential equations
with polynomial coefficients. Such functions appear frequently in
many applications, and efficient algorithms are available for answering
all sorts of questions about them~\cite{kauers23c}.

But not every function of interest belongs to the class of D-finite functions.
The tangent, the exponential generating function for Bernoulli numbers,
the ordinary generating function for partition numbers, the Weierstra\ss-$\wp$
function, Painleve transcendents, and Jacobi $\theta$-functions are prominent
examples of functions that are not D-finite.

However, these functions still belong to the class of D-algebraic functions.
For a function~$f$ to be D-algebraic means that there is a polynomial $P$
such that $P(f,f',\dots,f^{(r)})=0$, i.e., the defining differential equation
for~$f$ may be nonlinear.

D-algebraic functions have recently attracted increased interest in the context
of combinatorics. For example, the exponential generating function for labeled
trees was shown to be D-algebraic~\cite{bostan20a}. Also restricted lattice
walks~\cite{bernardi21}, Eulerian orientations~\cite{melou20}, colored
planar maps~\cite{bernardi17a}, and permutation patterns~\cite{chen24} lead to
D-algebraic functions that are not D-finite. 

At the same time, D-algebraic functions are also interesting from the perspective
of computer algebra. 
A specific D-algebraic function is uniquely determined by a differential
equation which it satisfies and some finitely many initial terms of its series
expansion.
Denef and Lipshitz~\cite{denef84,denef89} give an algorithm for
checking whether two D-algebraic functions given in this way are equal. More
recent work in this direction is due to van der Hoeven~\cite{hoeven19}.

Based on the classical theory of differential algebra~\cite{ritt50,kolchin73}, a
constructive elimination theory has been developed, see,
e.g.,~\cite{seidenberg56,ollivier91,boulier96,carraFerro97,carraFerro07,rueda16}
and the references therein. One consequence of this theory is that
the class of D-algebraic functions is closed under addition, multiplication,
division, and composition, and some other operations. 
Manssour et al.~\cite{AEMaSaTe25} recently proposed new algorithms for executing
such closure properties. Given defining differential equations for two D-algebraic
functions $f$ and~$g$, these algorithms compute defining differential equations
for $f+g$, $f\cdot g$, $f/g$, $f\circ g$, etc.

It is not difficult to see why the class of D-algebraic functions is closed
under these operations if we assume that the functions and their derivatives can
be identified with elements of a field. If $f$ satisfies a polynomial
differential equation $P(f,f',\dots,f^{(r)})=0$, then $f^{(r)}$ is algebraic
over the field generated by $f,f',\dots,f^{(r-1)}$, so this field has a
transcendence degree of (at most)~$r$. Note that this field is closed under
differentiation.
If $g$ is another D-algebraic function satisfying a differential equation of order~$s$,
so that the field generated by $g$ and its derivatives has transcendence degree
(at most)~$s$, then there is a field of transcendence degree (at most) $r+s$
containing, say, the sum~$h=f+g$ and all its derivatives. This implies that
$h,h',\dots,h^{(r+s)}$ are algebraically dependent, and therefore that $h$
satisfies an equation of order at most $r+s$.

Besides confirming the closure under addition, this argument suggests an algorithm
for finding differential equations of sums, products, etc. of given D-algebraic
functions, and it provides bounds for the orders of these equations. All this is
not too different from closure properties for D-finite functions. The main difference
is that nonlinear elimination theory has to be employed in place of linear algebra.

As far as D-finite functions are concerned, not only bounds on the orders of the
resulting equations are known but we also have bounds on the degrees of the
polynomial coefficients~\cite{kauers14f,kauers23c}. The combination of order and
degree gives a more realistic idea of how big an equation is. The purpose of the
present paper is to provide analogous results for D-algebraic functions.

More precisely, rather than obtaining bounds on the size of the coefficients, we derive
bounds on the total degree of the polynomials~$P$. Our bounds apply under the assumption
that certain ideals are sufficiently generic. The bounds are quite large. Although
better bounds are available for more specific situations (see, e.g., \cite{mukhina25}),
we believe that for the generic case this is not due to pessimistic overestimation but
an indication that closure properties for D-algebraic functions can indeed lead to rather
large equations.

This may be an explanation why for some of the D-algebraic functions that have recently
come up in combinatorics, we do not explicitly know their defining equations even though
they could in principle be obtained from rather simple constituents by applying closure
properties. Perhaps they are simply too big.

For D-finite functions, it has been pointed out that a slight increase in order can allow
for a drastically smaller degree. This observation has led to the concept of order-degree
curves~\cite{bostan07,chen12c,jaroschek13a,kauers14f,kauers23c}. As we shall see, there
is a similar phenomenon in the nonlinear case. 

\section{Dimensions}\label{section_hilbert}
Throughout this paper, if $R$ is a commutative ring and $P_1,\dots,P_n$ are
elements of $R$ we will denote by $\langle P_1,\dots,P_n\rangle$ the ideal
$P_1 R+\dots+P_nR$ generated by $P_1,\dots,P_n$ in~$R$.

We recall some facts about Hilbert series, Hilbert polynomials and
dimensions of algebraic varieties that will be used later.
Throughout this section let $K$ be a field and $R=K[s,x_1,\dots,x_n]$.

\begin{definition}\label{def_hilbert_function}
  Let $I$ be a homogeneous ideal of~$R$.
  For $i\in\set N$, let $R_i$ be the $K$-vector space of all homogeneous
  polynomials of degree~$i$ (together with zero), and let $I_i=I\cap R_i$.
  \begin{enumerate}
  \item The function $HF_I\colon\set N\to\set N$, $HF_I(i)=\dim_K(R_i/I_i)$
    is called the \emph{Hilbert function} of~$I$.
  \item The generating series
    \[
    HS_I(t):=\sum_{i=0}^\infty HF_I(i)t^i\in \mathbb{Z}[[t]]
    \]
    is called the \emph{Hilbert series} of~$I$.
  \end{enumerate}
\end{definition}

\begin{proposition} \cite[\S9.3]{CoLiOS07}\label{prop:hilbertpoly}
  For every homogeneous ideal $I$ of $R$ there
  exists a polynomial $HP_I\in\set Q[t]$ such that
  $HP_I(i)=HF_I(i)$ for all sufficiently large~$i\in\set N$.
\end{proposition}

\begin{definition}
  Let $I$ be a homogeneous ideal of~$R$.
  \begin{enumerate}
  \item The polynomial $HP_I$ from Prop.~\ref{prop:hilbertpoly} is called
    the \emph{Hilbert polynomial} of~$I$.
  \item The \emph{(Hilbert) dimension} of $I$ is defined as 
  $\dim I:=\deg(HP_{I})$.
  \end{enumerate}
\end{definition}

\begin{definition}
  We say that a chain of prime ideals $\mathfrak{p}_0\subsetneq
  \mathfrak{p}_1\subsetneq\dots\subsetneq\mathfrak{p}_n$ has length~$n$. If
  $I$ is an ideal of $K[x_1,\dots,x_n]$, the Krull dimension of~$I$,
  denoted $\kdim (I)$, is the maximum of the lengths of the chains
  of prime ideals containing $I$.
\end{definition}
  
  If $I$ is an ideal in $K[x_1,\dots,x_n]$ then its Krull dimension is also
  the Hilbert dimension of its homogenization in $R$. Conversely, the
  Hilbert dimension of a homogeneous ideal is one less than its Krull
  dimension. In particular if $I$ is a homogeneous prime ideal in $R$, then its
  Hilbert dimension is one less than the transcendence
  degree of $\mathrm{Frac} \left(\nicefrac{R}{I}\right)$ over~$K$.
  If $I$ is homogeneous but not prime, then its Hilbert dimension is the maximum
  Hilbert dimension of the prime ideals that contain it.

\begin{proposition}\label{prop_HS_nonzerodiv}
  Let $I$ be a homogeneous ideal of $R$ and let
  $P\in R$ be a homogeneous polynomial of degree~$d$.
  Then
  \[
  (\dim I)-1\leq\dim(I+PR)\leq\dim I.
  \]
  Furthermore if $P$ is not a zero
  divisor in $\nicefrac{R}{I}$ then
  \[
  HS_{I+\<P>}(t)=(1-t^d)HS_{I}(t).
  \]
  In particular, $\dim(I+\<P>)=(\dim I)-1$.
\end{proposition}

\begin{proof}
  A proof that $(\dim I)-1\leq \dim (I+\<P>)\leq\dim I$ can be found in
  \cite[section~9.4~Theorem~3]{CoLiOS07}.
  For all $i\in\mathbb{N}$, let $R_i\subseteq R$ and $I_i\subseteq I$ be as defined in
  Definition~\ref{def_hilbert_function}.
  Consider the map $m_i\colon R_i/I_i\to R_{i+d}/I_{i+d}$, $m_i(Q)=PQ$.
  If $P$ is not a zero divisor in~$R/I$, then $m_i$ is injective, so
  \[
  \dim\im m_i=\dim(R_i/I_i)=HF_I(i)
  \]
  and
  \begin{alignat*}1
    \dim\coker m_i
    &=\dim (R_{i+d}/I_{i+d})/\im m_i\\
    &= \dim R_{i+d}/(I_{i+d} + P I_i)\\
    &= HF_{I + \<P>}(i+d).
  \end{alignat*}
  It follows that
  \begin{alignat*}1
    HF_I(i+d)
    &= \dim R_{i+d}/I_{i+d}\\
    &= \dim \im m_i + \dim\coker m_i\\
    &= HF_I(i) + HF_{I + \<P>}(i+d)
  \end{alignat*}
  for every~$i$. This implies the claim about $HS_{I+\<P>}$.

  For sufficiently large $i$, we have $HF_{I}(i)=HP_{I}(i)$, thus
  \[
  HP_{I+\<P>}(i)=HP_I(i)-HP_I(i-d),
  \]
  hence $HP_{I+\<P>}$ is a polynomial of degree exactly $\dim I-1$.
\end{proof}

This proposition implies that a homogeneous ideal generated by $r$ elements
in $K[s,x_1,\dots,x_n]$ is of dimension at least $n-r$. If we have an
equality, the corresponding projective variety is called a \emph{complete intersection.}

\begin{definition} Let $(P_1,\dots,P_k)$ be a tuple of homogeneous
  polynomials in~$R$.
  \begin{enumerate}
  \item $(P_1,\dots,P_k)$ is called a \emph{complete intersection}
    if $P_1,\dots,P_k$ generate an ideal of (Hilbert) dimension $n-k$.
  \item $(P_1,\dots,P_k)$ is called a \emph{regular sequence}
    if $P_{i+1}$ is not a zero divisor in
    $\nicefrac{R}{\langle P_1,\dots,P_i\rangle}$ for any $1\leq i<k$.
  \end{enumerate}
\end{definition}

It should be noted that according to this definition, if $(P_1,\dots,P_k)$
is a complete intersection, then the projective variety corresponding to the
ideal $\<P_1,\dots,P_k>$ is a complete intersection.
However, the converse is not true: the ideal $\langle P_1,\dots,P_k\rangle$
might be of dimension strictly greater than $n-k$ and nevertheless admit
a smaller basis. 

\begin{proposition}\label{prop_HS_regular_sequence}
  Let $(P_1,\dots,P_k)$ be a regular sequence of homogeneous polynomials in~$R$,
  and let $d_1,\dots,d_k$ be their respective degrees. Then
  \[
  HS_{\langle P_1,\dots,P_r\rangle}(t)=\frac{(1-t^{d_1})\cdots(1-t^{d_k})}{(1-t)^n}
  \]
  and $\dim\langle P_1,\dots,P_k\rangle=n-k$.
\end{proposition}
\begin{proof}
  The proof is by induction on~$k$. If $k=0$ then $HF_R(i)$ is the
  number of monomials of degree $i$ which is
  $\binom{n+i}{n}$. Thus $HS_R(t)=(1-t)^{-n}$.
  For the induction step, apply Proposition~\ref{prop_HS_nonzerodiv}.
\end{proof}

Thus we can precisely know the Hilbert function of an ideal generated by a
regular sequence, which will be useful in the later sections. It is obvious
that regular sequences are complete intersections. The following proposition
shows that this is in fact an equivalence.

\begin{proposition}
  For a tuple $(P_1,\dots,P_k)$ of homogeneous polynomials in~$R$,
  the following properties are equivalent:
  \begin{enumerate}
  \item $(P_1,\dots,P_k)$ is a complete intersection
  \item $(P_1,\dots,P_i)$ is a complete intersection for every $i=1,\dots,k$.
  \item For every $i\in\{1,\dots,k\}$, any minimal prime ideal
    containing $\langle P_1,\dots,P_i\rangle$ is of dimension $n-i$.
  \item $(P_1,\dots,P_k)$ is a regular sequence.
  \end{enumerate}
\end{proposition}
\begin{proof}
  Proposition~\ref{prop_HS_regular_sequence} shows that $(4)\Rightarrow(1)$.

  We show $(1)\Rightarrow(2)$ by descending induction on~$i$.
  If the tuple $(P_1,\dots,P_{i+1})$ is a complete intersection then
  $\dim \langle P_1,\dots,P_{i+1}\rangle=n-i-1$.
  \[
  \dim \langle P_1,\dots,P_{i+1}\rangle\geq \dim \langle P_1,\dots,P_{i}\rangle-1.
  \]
  Thus $\dim \langle P_1,\dots,P_{i}\rangle\leq n-i$ and since the ideal
  is generated by only $i$ elements this is an equality.
  
  We now show that $(2)\Rightarrow(3)$ by induction on~$i$. For $i=0$
  this is obvious. Suppose that every minimal prime ideal containing $\langle
  P_1,\dots,P_i\rangle$ is of dimension $n-i$. Let $\mathfrak{p}$ be a
  minimal prime ideal containing $\langle P_1,\dots,P_{i+1}\rangle$. Then
  $\dim \mathfrak{p}\leq \dim\langle P_1,\dots,P_{i+1}\rangle=n-i-1$. By
  induction hypothesis this means that $\mathfrak{p}$ is not a minimal
  ideal containing $\langle P_1,\dots,P_i\rangle$ so there exists $\mathfrak{q}$ such that
  $\langle P_1,\dots,P_i\rangle\subset \mathfrak{q}\subset\mathfrak{p}$. Thus $\mathfrak{p}$ is a minimal prime
  ideal containing $\mathfrak{q}+P_{i+1}R$. We know from
  \cite[Theorem~1.23]{Shafarevich94} that this implies that $\dim
  \mathfrak{p}\geq \dim \mathfrak{q}-1=n-i-1$.

  Let us now suppose that $(3)$ is true and show
  that $(P_1,\dots,P_r)$ is a regular sequence. Suppose that we have shown
  that $(P_1,\dots,P_i)$ is a regular sequence and show that $P_{i+1}$ is
  not a zero divisor in $\nicefrac{R}{\langle P_1,\dots,P_i\rangle}$. But
  $P_{i+1}$ can only be such a zero divisor if $P_{i+1}$ belongs to some
  minimal prime ideal $\mathfrak{p}$ containing $\langle
  P_1,\dots,P_i\rangle$. But then $\dim \langle
  P_1,\dots,P_{i+1}\rangle\geq \dim \mathfrak{p}=n-i$ by $(3)$. This
  cannot be the case as $\dim \langle P_1,\dots,P_{i+1}\rangle = n-i-1$.
\end{proof}

\section{Setting}\label{section_setting}

Let $K$ be a differential field.
This means that the field $K$ is equipped with a map $D\colon K\to K$ satisfying
$D(a+b)=D(a)+D(b)$ and $D(ab)=D(a)b+aD(b)$ for all $a,b\in K$.
Such a map is called a derivation on~$K$.
We will also use the notations $a',a'',a'''$, and $a^{(k)}$ instead of
$D(a),D^2(a),D^3(a)$, and $D^k(a)$, respectively. 
An element $c$ of $K$ is called a constant if $c'=0$. 
We denote by $C\subseteq K$ the subset of all constants of~$K$. This 
set is actually a subfield of~$K$.

Typical choices for our considerations are $K=\set Q(x)$ with $x'=1$ or $K=\set Q$.
In both cases, we have $C=\set Q$.

We shall consider functions $f_1,\dots,f_n$ that belong to a certain field~$F$
that is closed under differentiation and contains (an isomorphic copy of)~$K$.
It does not matter where the functions are
defined, but it does matter that we can view them as elements of a differential
field. 

For every $r_1,\dots,r_n\in\set N$, consider the polynomial ring $R_{r_1,\dots,r_n}$
whose coefficient field is $K$ and which has $r_1+\cdots+r_n+n$ variables that we denote by
\begin{alignat*}1
  &y_1,y_1',\hbox to4em{\dotfill}, y_1^{(r_1)},\\
  &y_2,y_2',\hbox to4.5em{\dotfill}, y_2^{(r_2)},\\
  &\vdots\\
  &y_n,y_n',\hbox to3.5em{\dotfill}, y_n^{(r_n)}.
\end{alignat*}
The naming of the variables is chosen such as to suggest a way to differentiate polynomials:
The derivative of an element of $R_{r_1,\dots,r_n}$ is defined as the element of
$R_{r_1+1,\dots,r_n+1}$ obtained by differentiating according to the usual rules for
differentiation, the derivation of~$K$, and the rules $(y_i^{(j)})'=y_i^{(j+1)}$.

We have $R_{r_1,\dots,r_n}\subseteq R_{r_1',\dots,r_n'}$ whenever $r_i\leq r_i'$ for all~$i$.
The \emph{order} of an element $P$ of $R_{r_1,\dots,r_n}$ with respect to
$y_i$ is the smallest
$k$ such that $P$  does not contain any of the variables $y_i^{(l)}$ for
$l>k$. It is denoted by $\ord_i(P)$. Note that if $P$ is independent from
$y_i$ and its derivative we find $\ord_i(P)=0$. This specific point may be
open to debate, but will not matter in the rest of this paper.
The order of $P$ is the smallest $k$ such that $P$ is contained in $R_{k,k,\dots,k}$.
Note that $\ord(P)=\max_{i=1}^n\ord_i(P)$.

Recall that $F$ is a differential field extension of $K$ which contains the
$f_i$. There exists a unique ring homomorphism
$\phi\colon R_{r_1,\dots,r_n}\to F$ which maps $y_i$ to $f_i$ and $K$ to
itself such that $\phi(P')=\phi(P)'$ for every $P\in
R_{r_1-1,\dots,r_n-1}$. Its kernel is the ideal of all algebraic relations
among $f_1,\dots,f_n$ and their derivatives up to respective
orders~$r_1,\dots,r_n$. If there is just one function
($n=1$), then for this function to be D-algebraic means that the kernel is nonzero for
sufficiently large~$r_1$. Its elements amount to differential equations satisfied by the function.

If there are several functions, we assume that for some $r_1,\dots,r_n$ we know (generators of) an
ideal~$I$ of $R_{r_1,\dots,r_n}$ that is contained in~$\ker\phi$. Typically we will not know if
$I=\ker\phi$, but we shall assume that $I$ is sufficiently large to guarantee that the functions
under consideration all are D-algebraic. This is the essence of part~2 of the following definition.

\begin{definition}\label{def:D-algebraic}
  \begin{enumerate}
  \item If $I$ is an ideal of $R_{r_1,\dots,r_n}$, then we write $I'$ for the ideal
    of $R_{r_1+1,\dots,r_n+1}$ generated by the elements of $I$ and their
      first derivatives.
  \item An ideal $I$ of $R_{r_1,\dots,r_n}$ is called \emph{D-algebraic} with
    respect to $y_l$ if there exists an $m\in\set N$ such that
    \[
    I^{(m)}\cap K[y_l,y_l',\dots,y_l^{(r_l+m)}]\neq\{0\}.
    \]
  \item An ideal $I$ of $R_{r_1,\dots,r_n}$ is called \emph{D-algebraic} if it is
    \emph{D-algebraic} with respect to all variables.
  \end{enumerate}
\end{definition}
If an ideal $I$ is D-algebraic with respect to $y_l$, then the elements of the
elimination ideal $I^{(m)}\cap K[y_l,y_l',\dots,y_l^{(r_l+m)}]$ amount to
differential equations satisfied by~$f_l$. In particular, an ideal is
D-algebraic if and only if all coordinates of all solutions (in all
sufficiently large differential field
extensions of $K$) are
D-algebraic.

For $r_1=\cdots=r_n=\infty$, we recover classical notions from the theory of differential
algebra. In this case, $R_{r_1,\dots,r_n}$ is the differential ring of differential
polynomials, $\phi$~is a differential homomorphism, an ideal $I$ of $R_{r_1,\dots,r_n}$
that is closed under differentiation is a differential ideal, and an ideal
is D-algebraic (with respect to all variables)
if and only if its differential dimension is zero. However, we will
mostly need to operate with the finite $r_1,\dots,r_n$.

We will sometimes prefer to work with homogeneous polynomials. We then use $s$ as homogenization variable
and write $R_{r_1,\dots,r_n}^h$ for the polynomial ring over $K$ whose variables are $s$ and
$y_i^{(j)}$ for $i=1,\dots,n$ and $j=0,\dots,r_i$. For a polynomial $P\in R_{r_1,\dots,r_n}$,
we write $h(P)\in R_{r_1,\dots,r_n}^h$ for its homogenization with $s$ as homogenization variable,
and for an ideal $I$ of $R_{r_1,\dots,r_n}$ we write $h(I)$ for the ideal of $R_{r_1,\dots,r_n}^h$
generated by all $h(P)$ with $P\in I$. Note that we have $h(I')=h(I)'$, i.e., the homogenization
variable behaves like a constant.

\begin{definition}
Let $P_1,\dots,P_n\in R_{r_1,\dots,r_n}$ and let $r=\sum_{i=1}^n r_i$. 
  \begin{enumerate}
    \item Let $\rho\geq 0$. The tuple $(P_1,\dots,P_n)$ is called
      \emph{D-regular} at order $\rho$ if the tuple $(h(P_j)^{(k)})_{1\leq
      j\leq n,0\leq k\leq \rho}$ is a complete intersection.
    \item The tuple $(P_1,\dots,P_n)$ is called \emph{D-regular} with
      respect to $y_l$ if it is D-regular at order $r-r_l$.
  \end{enumerate}
\end{definition}


\section{Degree bounds in complete intersections}

We consider a tuple $(P_1,\dots,P_n)$ of elements of $R_{r_1,\dots,r_n}$,
where the $r_i$ are chosen as small as possible, and let $I=\<P_1,\dots,P_n>$
be the ideal they generate. We assume that this ideal is D-algebraic in
the sense of Def.~\ref{def:D-algebraic}.

Note that $I$ might not be D-algebraic even if $(h(P_1),\dots,h(P_n))$ is a complete
intersection. For example, for
\begin{alignat*}1
  &P_1=y_1''-2y_1y_1',\\
  &P_2=(y_1'-y_1^2)y_2'^2-y_1''(y_1'-y_1^2)y_2
\end{alignat*}
we have that $(h(P_1),h(P_2))$ is a complete intersection, but one can 
check that for any $c\in C$, $(c-x)^{-1}$ is a solution of~$P_1$, but
also of $y_1'-y_1^2$. Therefore, $P_2((c-x)^{-1},y_2)=0$ regardless of~$y_2$.

The goal of this section is to determine bounds on the degree of a
nonzero element in $I^{(m)}\cap K[y_i,y_i',\dots]$. Note that it is
in general not obvious for which $m$ this is true, even if bounds on
the order of the solutions are known.

\begin{example}\label{example_non_complete_intersection}
Consider the system defined by $P_1=y_1y_1''-y_1'^2$ and $P_2=(y_2-y_1)^2+(y_2'-y_1')^4$. If $(f_1,f_2)$
is a solution of $(P_1,P_2)$ then $f_2$ is actually the sum of two
D-algebraic function, $f_1$ and~$f$, where $f$ is a solution of
$y^2+y'^4=0$. This is how the example was presented in
\cite[Example~4.2]{AEMaSaTe25}. Thus $f_2$ lies in a differential field extension of
$\mathbb{Q}$ of transcendence degree $3$ and is thus solution of a
D-algebraic equation of order~$3$. However, the ideal $\langle
P_1,P_2\rangle^{(2)}$ does not have a non
trivial intersection with $\mathbb{Q}[y_2,y_2',\dots]$, as was stated in
  \cite{AEMaSaTe25}.

A closer look at the solutions of the differential
equations reveals that they can, in this example, be written in closed form. The
solutions of $y^2+y'^4$ are $0$ and polynomials of the form
\[
\frac14(-1)^k i x^2+ax+(-1)^{k+1}ia^2
\]
  with $k\in\{0,1\}$, $a\in C$ and $i^2=-1$.
All of those solutions satisfy the equation $(-1)^kiy'^2+y=0$ of
order~$1$ and of degree~$2$ rather than only an equation of degree~$4$.
Likewise, the solutions of $P_1$ are
exponential functions of the form $\lambda\exp(cx)$ ($\lambda,c\in C$) and
satisfy an equation of order and degree~$1$.

  If we were to fix $k$ and $c$ and take $Q_1=y_1'-cy_1$ and
  $Q_2=(-1)^ki(y_2'-y_1')^2+(y_2-y_1)$ instead of $P_1$ and $P_2$ and
  $J=\langle Q_1,Q_2\rangle$ instead of $I$, we find that
  $J'\cap\mathbb{C}[y_2,y_2',y_2'']\neq\{0\}$.
\end{example}


\begin{theorem}\label{main_theorem}
  Let $P_1,\dots,P_n\in R_{r_1,\dots,r_n}$ and let $l\in\{1,\dots,n\}$.
  Suppose that $r_i=\max_{j=1}^n\ord_i(P_j)$ for all $i\in\{1,\dots,n\}$.
  In addition, we assume that for each $i$, at least one the $P_j$ is not
  independent from $y_i$ or its derivatives.
  Let $d:=\prod_{j=1}^n\deg P_j$.

  Let $r_{\min}=\sum_{i=1}^n r_i$, $r\geq r_{\min}$
  $(P_1,\dots,P_n)$ is \emph{D-regular} at order $r-r_l$. Then the elimination ideal
  \[
  \langle P_1,\dots ,P_n\rangle^{(r-r_l)}\cap
  K[y_l,y_l',\dots]
  \]
  contains a nonzero element of order $r$ and degree~$k$ as soon as
  $k>(r+1)(d^{1+(r_{\min}-r_l)/(r-r_{\min}+1)}-1)$.
\end{theorem}
\begin{proof}
  First note that for any non constant polynomial $P\in R_{r_1,\dots,r_n}$
  we have $\deg(P')=\deg(P)$. For each
  $i\in\{1,\dots,n\}$ we set $d_i=\deg P_i$. Let
  $I:=h(\<P_1,\dots,P_n>)\subset R^h_{r_1,\dots,r_n}$.

  We know from Proposition~\ref{prop_HS_regular_sequence} that
  \begin{alignat*}1
    HS_{I^{(r-r_l)}}(t)&=
    \frac{\prod_{i=1}^{n}(1-t^{d_i})^{r-r_l+1}}{(1-t)^{r_{\min}+n(r-r_l+1)}}\\
    &=(1-t)^{-r_{\min}}\prod_{i=1}^n(1+t+\dots +t^{d_i-1})^{r-r_l+1}.
  \end{alignat*}
  We claim for any sequence of stricly positive integers
  $(u_n)_{n\in\mathbb{N}^*}\in (\mathbb{N}^*)^{\mathbb{N}^*}$,
  if we write
  \[(1-t)^{-r_{\min}}\prod_{i=1}^n(1+t+\dots+t^{u_i-1})=\sum_{k=0}^\infty
  a_{n,k}t^k\]
  then $a_{n,k}\leq \binom{r_{\min}+k}{k}\prod_{i=1}^n u_i$. This is obviously true
  for $n=0$. Then if the result is true for $n$ then
  \[
    (1-t)^{-r_{\min}}\prod_{i=1}^{n+1}(1+t+\dots+t^{u_i-1})=(1+t+\dots+t^{u_{n+1}-1})\sum_{k=0}^\infty
  a_{n,k}t^k
\]
  Thus \begin{align*}
    a_{k,n+1}&=\sum_{j=0}^{u_{n+1}-1}a_{k-j,n}\\
    &\leq \left(\prod_{j=0}^n u_j\right)\sum_{i=0}^{u_{n+1}-1}\binom{r_{\min}
    +k-j}{k-j}\\
    &\leq u_{n+1}\left(\prod_{j=0}^n u_j\right)\binom{r_{\min}+k}{k}
  \end{align*}
  which proves the statement by induction on $n$. It follows that
  \[
    HF_{I^{(r-r_l)}}(k)\leq
    d^{r-r_l+1}\binom{r_{\min}+k}k.
  \]
  The space $V_k\subseteq K[s,y_l,\dots,y_l^{(r)}]$ of homogeneous polynomials
  of degree $k$ has dimension $\binom{r+1+k}k$ over~$K$. By the
  definition of the Hilbert polynomial, its image in
  \[
  R^h_{r+r_1,\dots,r+r_n}/I^{(r-r_l)}
  \]
  under the natural morphism is a vector space of dimension at most
  $d^{r-r_l+1}\binom{r_{\min}+k}k$.

  If $k> (1+r)(d^{1+(r_{\min}-r_l)/(r-r_{\min}+1)}-1)$ then we have
  \[
    \prod_{i=1}^{1+r-r_{\min}}\frac{r_{\min}+i+k}{r_{\min}+i}
    >\left(1+\frac{k}{r+1}\right)^{1+r-r_{\min}}\geq d^{r-r_l+1}
  \]
  and so $\binom{r+1+k}{r+1}>HP_{I^{(r-r_l)}}(k)$.
  This means that $V_k$ contains nonzero polynomials that are mapped to zero.
  By setting $s=1$, any such element translates into a nonzero element of
  $$\langle P_1,\dots ,P_n\rangle^{(r-r_l)}\cap K[y_l,y_l',\dots]$$ of the
  announced order and degree.
\end{proof}

In view of the exponential size of the bound of Theorem~\ref{main_theorem},
we were not able to check experimentally how tight it is. The required
computations were too large. However, to at least get some idea, we
carried out some experiments for a similar, though different problem.
Given $n+1$ polynomials $P_0,\dots,P_n$ in $K[x_1,\dots,x_n]$ of degree~$d$,
it is clear that they must be algebraically dependent. What is the typical
degree of their algebraic relation? A calculation similar to the proof of
Theorem~\ref{main_theorem} shows that there is an algebraic relation of
total degree $k$ as soon as $\binom{n+1+k}{n+1}>\binom{n+kd}{n}$.
This is true for $k\geq (n+1)(d^n-1)$. However, experiments suggest
that an algebraic relation already exists for $k\geq d^n$. We do not
know the reason for this discrepancy, but it suggests that bound of
Theorem~\ref{main_theorem} perhaps also overshoots by a factor of~$r+1$.

The hypothesis that the family $(P_1,\dots,P_n)$ is D-regular (for any
variable) is not trivial in general, even
if the ideal $\langle P_1,\dots, P_n\rangle$ is D-algebraic (with respect to any
variable), as
shown in Example~\ref{example_non_complete_intersection}. Outside of the
differential context, a generic family of polynomials is a complete
intersection, similarly to generic intersections of hyperplanes.
The family of polynomials considered here however is not random as it is
composed of the successive derivative of given differential polynomials.
Nevertheless, experiments conducted on random operators of small orders
and degrees seem to indicate that this hypothesis is often satisfied.

Theorem~\ref{main_theorem} also shows that, in the case of complete
intersection, going to a higher derivative order may provide equations of
smaller degrees. This phenomenon is well-known in the case of linear differential
operators~\cite{jaroschek13a} and here finds its nonlinear counterpart. 
Two things should be noted however. The first is that unlike in the linear case,
the ``order-degree curve'' that we obtain here is increasing for big enough~$r$.
This is an artifact of the approximations used during the proof of the theorem. The
second is that, unlike for linear differential operators, the size of a polynomial
does not linearly depend on its order and its degree. It would therefore be more relevant to compare 
how the total number of monomials depends on the order.
We have conducted tests for a few values of $d$, $r_l$ and $r_{min}$, whose
results are presented in Figure~\ref{figure1}. The graphs in
Figure~\ref{figure1} show the evolution along $r-r_{\min}$ of the number of monomials of
degree $k$ in $K[s,y_l,y_l',\dots,y_l^{r+r_l}]$ for the smallest $k$ for
which this number is strictly bigger than
$d^{r-r_l+1}\binom{r_{\min}+k}{k}$, at which point we can ensure the
existence of a nontrivial element in the intersection ideal.  Those tests suggest that the
number of possible monomials drops significantly for the first few values
of $r>r_{\min}$. It should be noted that the number of monomials
presented here only results from comparing the number of equations and
the number of variables in the linear system considered in the proof of Theorem~\ref{main_theorem}. 
It was not obtained by actually solving these linear systems, for they are
too big to handle, so these curves may overshoot.
\begin{figure}
  \includegraphics[width=0.9\linewidth]{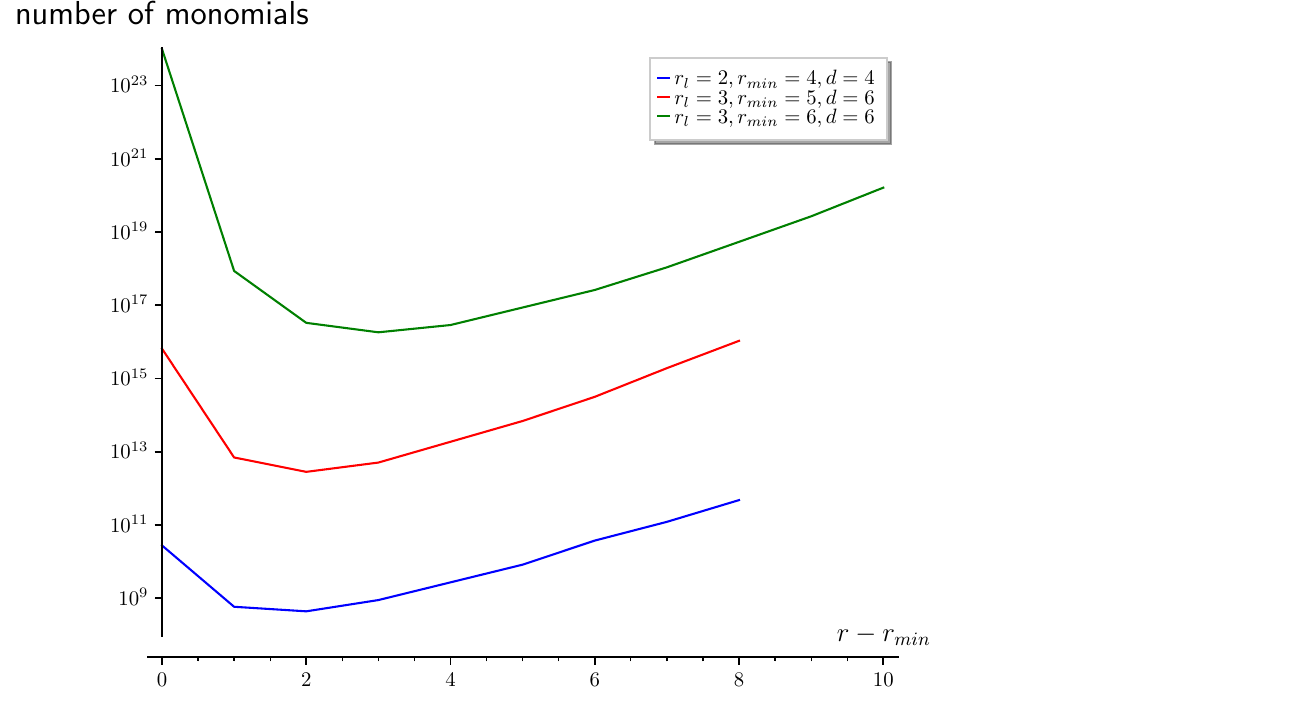}
  \caption{order-number of monomials curves}
  \label{figure1}
\end{figure}

\section{Corollaries on degree bounds for algebraic operations}

\begin{proposition}\label{plus_times}
  Let $f_1,\dots,f_n$ be D-algebraic functions, as well as $P_1,\dots,P_n\in
  R_{r_1,\dots,r_n}$ such that $P_i\in K[y_i,y_i',\dots,y_i^{(r_i)}]$ for
  all $i\in\{1,\dots,n\}$ and
  $Q\in R_{r_1,\dots,r_n}$. We note $r_{\min}:=\sum_{i=1}^n
  r_i$ and $d:=\prod_{i=1}^n\deg(P_i)$. Let $r\geq
  r_{\min}$ and assume that
  \begin{itemize}
    \item $P_i(f_i)=0$ for all $i\in\{1,\dots,n\}$.
    \item  $P_i$ is D-regular at order $r$ for all $i\in\{1,\dots,n\}$.
  \end{itemize}
  Then $Q(f_1,\dots,f_n)$ is solution of a D-algebraic equation of order
  $r$ and of degree $k$ or less, as soon as
  \[
  k>(r+1)((\deg(Q)d)^{1+r_{\min}/(r-r_{\min}+1)}-1).
  \]
\end{proposition}
\begin{proof}
  It is enough to show that the family
  $(P_1,\dots,P_n,z-Q(y_1,\dots,y_n))\in R_{r_1,\dots,r_n}[z]$ is D-regular
  at order~$r$.
  
  Let $I:=\langle P_1,\dots,P_n,z-Q(y_1,\dots,y_n)\rangle \subset
  R_{r_1,\dots,r_n}[z]$ and $I_1:=\langle P_1,\dots,P_n\rangle\subset
  R_{r_1,\dots,r_n}$. 
  We know that $\dim h(I)^{(r)}=\kdim (I^{(r)})$. There is a natural morphism
  \[
  R_{r_1+r,\dots,r_n+r}/I_1^{(r)}\rightarrow
  R_{r_1+r,\dots,r_n+r}[z,z',\dots,z^{(r)}]/I^{(r)}.
  \]
  This morphism is surjective. Indeed,
  \[
  (z-Q(y_1,\dots,y_n))^{(k)}=z^{(k)}-Q(y_1,\dots,y_n)^{(k)}
  \]
  for all
  $k\leq r$. By successive euclidean divisions, it follows that any element
  of $R_{r_1+r,\dots,r_n+r}[z,z',\dots,z^{(r)}]/I^{(r)}$ can be represented
  by an element of $R_{r_1+r,\dots,r_n+r}$. It follows that
  \[
  \kdim (I^{(r)})\leq \kdim (I_1^{(r)}).
  \]  
  But we also have 
  \begin{alignat*}1
    R_{r_1,\dots,r_n}/I_1&\simeq K[y_1,\dots,y_1^{(r_1+r)}]/\langle P_1\rangle^{(r)}\\
    &\otimes_K K[y_2,\dots,y_2^{(r_2+r)}]/\langle P_2\rangle^{(r)}\\
    &\vdots\\
    &\otimes_K K[y_n,\dots,y_n^{(r_n+r)}]/\langle P_n\rangle^{(r)}.
  \end{alignat*}
  Thus we have
  \[
  \kdim (I_1)=\sum_{i=1}^n \kdim (\langle
  P_i^{(r)}\rangle)=\sum_{i=1}^n \dim h(\langle
  P_i\rangle)^{(r)}=\sum_{i=1}^n r_i=r_{\min}.
  \]
  Thus $\dim h(I)^{(r)}\leq r_{\min}$ and since it cannot be lower than this, the
  family $(P_1,\dots,P_n,z-Q(y_1,\dots,y_n))$ is D-regular at order $r$. We can now apply Theorem~\ref{main_theorem}.
\end{proof}

Proposition~\ref{plus_times} covers in particular the case of the addition
and multiplication of D-algebraic functions. The incorporation of divisions
requires stronger hypothesis.

\begin{proposition}\label{div}
  Let $f_1,\dots,f_n$ be D-algebraic functions, as well as $P_1,\dots,P_n\in
  R_{r_1,\dots,r_n}$ such that $P_i\in K[y_i,y_i',\dots,y_i^{(r_i)}]$ for
  all $i\in\{1,\dots,n\}$ and
  $Q_n,Q_d\in R_{r_1,\dots,r_n}$. We note $r_{\min}:=\sum_{i=1}^n
  r_i$ and $d:=\prod_{i=1}^n\deg(P_i)$. Let $r\geq
  r_{\min}$ and assume that
  \begin{itemize}
    \item $P_i(f_i)=0$ for all $i\in\{1,\dots,n\}$.
    \item The family $(P_1,\dots,P_n,Q_d z-Q_n)$ is D-regular at order $r$.
  \end{itemize}
  Then $Q_n(f_1,\dots,f_n)/Q_d(f_1,\dots,f_n)$ is solution of a D-algebraic equation of order
  $r$ and of degree $k$ or less as soon as
  \[k>(r+1)((\max(\deg(Q_n),\deg(Q_d))d)^{1+r_{\min}/(r-r_{\min}+1)}-1).\]
\end{proposition}
\begin{proof}
  This is a direct consequence of Theorem~\ref{main_theorem}
\end{proof}


\section{Bounds for the composition of D-algebraic functions}
When D-algebraic functions are indeed functions (for example meromorphic
functions), rather than abstract elements of a differential field, one might be tempted to consider the
composition operation. 
Another setting in which the composition operation is sometimes well defined
is that of power series.

It is known~\cite{AEMaSaTe25} that in both cases, the composition of two D-algebraic
functions is itself D-algebraic when this composition is well defined.
However, it is not completely
clear how to define the composition of two elements of an abstract
differential field. From an algebraic standpoint we want the composition on the right by a
given function to preserve algebraic relations. This
means that if $f_{1,0},\dots f_{1,n},f_2$ are ``functions'' such that the
compositions $f_{1,i}\circ f_2$ are well defined for all $i$, for any
algebraic relations $P(f_{1,0},\dots,f_{1,n})=0$ we must have $P(f_{1,0}\circ
f_2,\dots,f_{1,n}\circ f_2)=0$. Another way of saying this is that there
would be a ring homomorphism
\[
K[f_{1,0},\dots,f_{1,n}]\rightarrow K[f_{1,0}\circ f_2,\dots,f_{1,n}\circ f_2].
\]
If we want to define the
composition of a function $f_1$ with $f_2$, this must in particular apply
to the successive derivatives of $f_1$, $f_{1,i}=f_1^{(i)}$ for all
$i\in\mathbb{N}$. From a differential standpoint we want the composition
to satisfy the usual derivation rule $(f_1^{(i)}\circ f_2)'=f_2'\cdot
\left(f_1^{(i+1)}\circ f_2\right)$. Following these ideas we propose the following
definition.

\begin{definition}\label{def_composition}
  Let $F$ be a differential field extension of $K$, $f_1,f_2\in
  F$. An element
  $h$ in some differential field extension $E$ of $F$
  is called a composition of $f_1$ with $f_2$ if 
  there exists a family $(h_i)_{i\in\mathbb{N}}\in
  E^{\mathbb{N}}$ satisfying
  \begin{itemize}
    \item $h_0=h$
    \item $h_i'=f_2'h_{i+1}$ for all $i\in\mathbb{N}$.
    \item There exists a (algebraic) homomorphism
      $K[f_1,f_1'\dots]\rightarrow E$ which maps $f_1^{(i)}$ to $h_i$ for
      all $i\in\mathbb{N}$.
  \end{itemize}
\end{definition}

  The $h_i$ represent the functions $f_1^{(i)}\circ f_2$. It should be
  noted that according to this definition, if $h_0$ is a composition of
  $f_1$ with $f_2$, then $h_i$ is also a composition of $f_1^{(i)}$ with
  $f_2$ according to the same definition.

\begin{proposition}
  Let $f_1,f_2$ be two D-algebraic functions and $P_i\in
  K[y_i,y_i',\dots,y_i^{(r_i)}]$ for $i\in\{1,2\}$ such that
  \begin{itemize}
    \item $P_i(f_i)=0$ for $i\in\{1,2\}$.
    \item $(P_1)$ (resp. $(P_2)$) is D-regular at order $r_2$ (resp.
      $(r_1)$).
    \end{itemize}
    We note $d_i=\deg(P_i)$ for $i\in\{1,2\}$.
    Then any composition of $f_1$ with $f_2$ is a solution of a D-algebraic
    equation of order $r_1+r_2$ and of degree smaller than $k$ as soon as
    $$k>(r_1+r_2+1)((r_1+r_2+1)!d_1^{r_2}d_2^{r_1}-1).$$ Furthermore, this
    equation does not depend on the choice of the composition.
\end{proposition}
\begin{proof}
  Let $h\in E$ be a composition of $f_1$ with $f_2$ and let
  $(h_i)_{i\in\mathbb{N}}$ be as in Definition~\ref{def_composition}. We claim that
  $$(h_0,h_1,\dots,h_{r_1+r_2},f_2,f_2',\dots,f_2^{(r_1+r_2)},h,h',\dots,h^{(r_1+r_2)})$$
  is a solution of 
  \begin{enumerate}[label=(\roman*)]
    \item $P_1^{(j)}(y_1,y_1',\dots,y_1^{(r_1+r_2)})$ for all
      $j\leq r_2$
    \item $P_2^{(j)}(y_2,y_2',\dots,y_2^{(r_1+r_2)})$ for all
      $j\leq r_1$
    \item $d^{j}(z-y_1)$ for all $j\leq r_1+r_2$, with $d$ being the
      derivation on the ring
      $K[(y_1^{(i)})_{i\in\mathbb{N}},(y_2^{(i)})_{i\in\mathbb{N}}
      ,(z^{(i)})_{i\in\mathbb{N}}]$
  given by $d(y_1^{(l)})=y_2'y_1^{(l+1)}$, and the usual derivation on
  $y_2^{(l)}$ and $z^{(l)}$ for all $l$
  \end{enumerate}
  all of them seen as polynomials in
  \begin{alignat*}1
    K[&y_1,\dots,y_1^{(r_1+r_2)},\\
      &y_2,\dots,y_2^{(r_1+r_2)},\\
      &z,\dots,z^{(r_1+r_2)}].
  \end{alignat*}
  $(ii)$ is obvious by hypothesis on
  $P_2$. All the polynomials in $(i)$ are vanishing operators for $f_1$.
  But since there exists a field morphism which sends $f_1^{(i)}$ to
  $h_i$, the $h_i$ must be roots of those polynomials too.
  Finally we know that $h-h_0=0$. Differentiating this expression gives that
  $h'-h_0'=h'-f_2'h_1=0$, which is to say that we find a root of
  $d(z-y_1)$. By induction we get the result.

  It must be noted that $d^i(z-y_1)$ is always of the form
  $z^{(i)}-Q_i(y_1,\dots,y_1^{(r_1+r_2)},y_2,\dots,y_2^{(r_1+r_2)})$ with
  $\deg(Q_i)=i+1$.
  Following the same line of reasoning as in the proof of
  Proposition~\ref{plus_times}, we show that this family of polynomials,
  once homogenised, is a complete intersection. Let $I$ be the ideal generated by this
  family of polynomials. Then as we did in the proof of
  Theorem~\ref{main_theorem}, we can show that $HF_{h(I)}(k)\leq
  (r_1+r_2+1)!d_1^{r_2}d_2^{r_1}\binom{r_1+r_2+k}{r_1+r_2}$. We consider the
  natural morphism $\varphi$ which maps elements of
  $$K[s,z,z',\dots,z^{(r_1+r_2)}]$$ to their equivalence class in
  $R^h_{r_1+r_2,r_1+r_2,r_1+r_2}/I$. The map $\varphi$ maps the space of
  homogeneous polynomials of degree $k$, which is of dimension
  $\binom{r_1+r_2+1+k}{r_1+r_2+1}$ onto a space of dimension
  $HF_I(k)$. We can check that for $k\geq
  (r_1+r_2+1)((r_1+r_2+1)!d_1^{r_2}d_2^{r_1}-1)$ the restriction of $\varphi$ to the
  space of homogeneous polynomials of degree $k$ must have a non trivial
  element in its kernel.
  
  Thus $I\cap K[z,z',\dots,z^{(r_1+r_2)}]$ has a nonzero
  element of degree at most~$k$ which is a vanishing operator
  for any composition of $f_1$ with~$f_2$.
\end{proof}

\section{Variable elimination in special cases}

In some special cases, it is possible to loosen the hypothesis on our
system of equations so that we don't need to use complete intersections
hypothesis. Some functions are easy enough to manipulate and we can
ensure the existence of operators satisfying the complete intersection
property. In addition, we can here make use of resultants
instead of the analysis conducted in Theorem~\ref{main_theorem}.

We first consider the case of the elimination of algebraic functions. To be
precise we consider an algebraic function $g$ over $C(x)$, where $C$ is the
constant field, and a D-algebraic
function $f$ satisfying an equation $$P(f,f',\dots,f^{(r)})=0$$ with
coefficients in $C(x)[g]$, and we want to recover an equation in
$C[x][y,y',\dots,y^{(r)}]$. We are interested in both the total degree of the
resulting equation in the variables $y,y',\dots,y^{(r)}$ as well as its
degree in $x$.
\begin{proposition}
  Let $g$ be an algebraic function over $C(x)$ and let $Q_g\in C[x,y_1]$ be
  the minimal primitive polynomial of $g$ over $C(x)$. Let $f$ be a
  D-algebraic function over $C(x)[g]$ and
  \[
  P\in C[x,y_1,y_2,y_2',\dots,y_2^{(r)}]
  \]
  be such that
  $P(x,g,f,f',\dots,f^{(r)})=0$. In addition, we suppose that $Q_g\nmid P$.
  Then $R=\Res_{y_1}(P,Q_g)\in C[x,y_2,\dots,y_2^{(r)}]$ is a
  D-algebraic equation for $f$. Let $d_x$, $d_{y_1}$, $d_{y_2}$ and $d$ denote the
  degree in $x$, degree in $y_1$, total degree in $y_2,y_2',\dots,y_2^{(r)}$
  and total degree functions respectively. Then
  \begin{enumerate}
    \item $d_{y_2}(R)\leq d_{y_1}(Q_g)d_{y_2}(P)$
    \item $d_x(R)\leq d_x(P)d_{y_1}(Q_g)+d_{y_1}(P)d_x(Q_g)$
    \item $d(R)\leq d(P)d_{y_1}(Q_g)+d_{y_1}(P)d(Q_g)$
  \end{enumerate}
\end{proposition}
\begin{proof}
  Since $Q_g\nmid P$ and $Q_g$ is irreducible, $P$ and $Q_g$ can have no
  common factor, which implies that $Q\neq 0$. Furthermore,
  \[
  (x,g,f,\dots,f^{(r)})
  \]
  is a root of both $P$ and $Q_g$, which implies
  that $R(x,f,\dots,f^{(r)})=0$. The degree bounds directly come from the
  fact that $R$ is the determinant of a Sylvester matrix with coefficients
  in $C[x,y_2,\dots,y_2^{(r)}]$. The first $d_{y_1}(Q_g)$ columns of this
  matrix are the coefficients of $P$ while the $d_{y_1}(P)$ last
  coefficients are the coefficients of $Q_g$ (which, in particular, are of
  total degree $0$ in $y_2,\dots,y_2^{(r)}$), which yields the result.
\end{proof}
We now turn to the elimination of hyperexponential
functions.
\begin{proposition}\label{elimination_hyp}
  Let $g$ be a hyperexponential function over $C(x)$ and let
  $\frac{g'}{g}=\frac{u}{v}$, with $u,v\in C[x]$ coprime.
  Let $f$ be a D-algebraic function over $C(x,g)$ and
  let $P\in C[x,y_1,y_2,y_2',\dots,y_2^{(r)}]$ be a polynomial which is
  primitive in $y^{(r)}$ and
  separable as an element of $C(x,y_1,y_2,\dots,y_2^{(r-1)})[y_2^{(n)}]$, such
  that $P(x,g,f,\dots,f^{(r)})=0$. We suppose that $P$ is not independent
  of $y_1$ so that there is something to do. Let $d_x$, $d_{y_1}$, $d_{y_2}$ and $d$ denote the
  degree in $x$, degree in $y_1$, total degree in $y_2,y_2',\dots,y_2^{(r+1)}$
  and total degree functions respectively. Then there exists $Q\in
  C[x,y_2,\dots,y_2^{(r+1)}]$ such that $Q(x,f,\dots,f^{(r)})=0$ and
  \begin{enumerate}
    \item $d_{y_2}(Q)\leq 2d_{y_1}(P)d_{y_2}(P)$
    \item $d_x(Q)\leq d_{y_1}(P)(2d_x(P)+\max(d_x(u),d_x(v)))$
    \item $d(Q)\leq d_{y_1}(P)(2d(P)+\max(d_x(u),d_x(v)))$
  \end{enumerate}
\end{proposition}
\begin{proof}
  The polynomial $P'$ belongs in
  $C[x,y_1,y_1',y_2,\dots,y_2^{(r+1)}]$ of degree $1$ in $y_2'$ and
  $P'(x,g,g',f,\dots,f^{(r+1)})=0$. Since $g'=g\frac{u}{v}$ we can set
  $$P_1=vP'(x,y_1,y_1\frac{u}{v},y_2,\dots,y_2^{(r+1)})\in
  C[x,y_1,y_2,\dots,y_2^{(r+1)}]$$ and have
  $P_1(x,g,f,\dots,f^{(r+1)})=0$. We set $Q=\Res_{y_1}(P,P_1)$ and
  claim that $Q\neq 0$. Indeed if that was the case then $P$ and $P_1$
  would share an irreducible factor $q\in C[x,y_1,y_2,\dots,y_2^{(r)}]$,
  and since $P$ is primitive in $y^{(r)}$, $q$ can not be independent of~$y_2^{(r)}$.
  Since
  $P_1=(v\partial_{y_2^{(r)}}P)y_2^{(r+1)}+R(x,y_1,y_2,\dots,y_2^{(r)})$ it
  follows that $q$ must be a common factor of both $P$ and
  $v\partial_{y_2^{(r)}}P$ which is impossible since $P$ is 
  separable as an element of
  $C(x,y_1,y_2,\dots,y_2^{(r-1)})[y_2^{(n)}]$.
  Thus $Q\neq 0$ and $$Q(x,f,\dots,f^{(r+1)})=0$$ since
  $(x,g,f,\dots,f^{(r+1)})$ is a common root to both $P$ and $P_1$.
  
  Once again, the degree bounds come from the fact that $Q$ is the
  determinant of a Sylvester matrix of size $2d_{y_1}(P)$. The first
  $d_{y_1}(P_1)=d_{y_1}(P)$ columns are composed of the coefficients of $P$ while the
  last $d_{y_1}(P)$ columns contain the coefficients of $P_1$. 
\end{proof}
  Any D-algebraic function $f$ over $K(x)$ is a solution of a D-algebraic
  equation satisfying the hypothesis of Proposition~\ref{elimination_hyp}.
  Indeed if $P$ is an equation of order $r$ for $f$ then its squarefree part is also
  an equation for $f$. Furthermore the $gcd$ of its coefficients as a polynomial in
  $y^{(r)}$ is either an equation for $f$, in which case we apply the same
  analysis on it, or we can divide $P$ by it and get a primitive equation
  for $f$. All of those operations provide polynomials of smaller degrees
  than $P$. Thus the degree bound given in
  Proposition~\ref{elimination_hyp} must always be true, even if $P$ does
  not satisfy the hypothesis of the proposition. However
  the resultant formula used in the proof might give zero in this case.

We have seen how to go from equations over an algebraic function field, or
over the field generated by a hyperexponential function, to an equation
with polynomial coefficients. We end this section by considering the
elimination of the $x$ variable as well. We assume that $K=C(x)$
where $C$ is a field of characteristic $0$. If $f$ is a D-algebraic
function over~$K$, we know that $f$ is also D-algebraic over~$C$. How can
one recover a D-algebraic equation over $C$ for $f$ from an equation over
$C(x)$?

\begin{proposition}\label{elimination_x}
  Let $f$ be a D-algebraic function over $K=C(x)$ and
  let $P\in C[x,y,y',\dots,y^{(r)}]$ be a polynomial which is primitive in
  $y^{(r)}$ and
  separable as an element of $C(x,y,y',\dots,y^{(r-1)})[y^{(n)}]$, such
  that $P(f)=0$. Let
  $d_x$ be the degree of $P$ as a polynomial in the variable $x$ and $d$ be
  its total degree as an element of $C(x)[y,\dots,y^{(r)}]$. Then $f$
  satisfies a D-algebraic equation of order $r+1$ and of degree at most $2d_xd$ which is
  either $P'$ if $P'$ is independent from $x$, or $\Res_x(P,P')$
  otherwise.
\end{proposition}
\begin{proof}
  If $P'$ is independent from $x$ then there is nothing to prove. Otherwise
  we claim that $\Res_x(P,P')$ is not $0$. Indeed, if such was the
  case then $P$ and $P'$ would have a common irreducible factor
  $q(x,y,y',\dots,y^{(r+1)})$. But since $P$ is only a polynomial in
  $C[x,y,\dots,y^{(r)}]$ that must also be the case of $q$. Furthermore,
  since $P$ is primitive, $q$ is not independent from $y^{(n)}$. Since
  \[
  P'=y^{(r+1)}\partial_{y^{(r)}}P+R(x,y,\dots,y^{(n)}),
  \]
  for some $R\in K[x,y,\dots,y^{(r)}]$, $q$ is a factor of $P$ and $\partial_{y^{(r)}}P$
  which can not be the case since $P$ is supposed separable in $y^{(r)}$.
  Then $(x,f,\dots,f^{(r+1)})$ is a root of both $P$ and $P'$ so
  $f,f',\dots,f^{(r+1)}$ must be a root of $\Res_x(P,P')$. The
  bound comes from the fact that $\Res_x(P,P')$ is the determinant
  of a square matrix of size $2d_x$ with coefficients of degree at most $d$
  in $C[y,y',\dots,y^{(r+1)}]$.
\end{proof}

\bibliographystyle{ACM-Reference-Format}
\balance
\bibliography{bibliographie}


\begin{thebibliography}{25}


\ifx \showCODEN    \undefined \def \showCODEN     #1{\unskip}     \fi
\ifx \showDOI      \undefined \def \showDOI       #1{#1}\fi
\ifx \showISBNx    \undefined \def \showISBNx     #1{\unskip}     \fi
\ifx \showISBNxiii \undefined \def \showISBNxiii  #1{\unskip}     \fi
\ifx \showISSN     \undefined \def \showISSN      #1{\unskip}     \fi
\ifx \showLCCN     \undefined \def \showLCCN      #1{\unskip}     \fi
\ifx \shownote     \undefined \def \shownote      #1{#1}          \fi
\ifx \showarticletitle \undefined \def \showarticletitle #1{#1}   \fi
\ifx \showURL      \undefined \def \showURL       {\relax}        \fi
\providecommand\bibfield[2]{#2}
\providecommand\bibinfo[2]{#2}
\providecommand\natexlab[1]{#1}
\providecommand\showeprint[2][]{arXiv:#2}

\bibitem[{Ait El Manssour} et~al\mbox{.}(2025)]%
        {AEMaSaTe25}
\bibfield{author}{\bibinfo{person}{Rida {Ait El Manssour}},
  \bibinfo{person}{Anna-Laura Sattelberger}, {and} \bibinfo{person}{Bertrand
  {Teguia Tabuguia}}.} \bibinfo{year}{2025}\natexlab{}.
\newblock \showarticletitle{D-algebraic functions}.
\newblock \bibinfo{journal}{\emph{Journal of Symbolic Computation}}
  \bibinfo{volume}{128} (\bibinfo{year}{2025}), \bibinfo{pages}{102377}.
\newblock
\showISSN{0747-7171}
\urldef\tempurl%
\url{https://doi.org/10.1016/j.jsc.2024.102377}
\showDOI{\tempurl}


\bibitem[Bernardi and Bousquet-M{\'e}lou(2017)]%
        {bernardi17a}
\bibfield{author}{\bibinfo{person}{Olivier Bernardi} {and}
  \bibinfo{person}{Mireille Bousquet-M{\'e}lou}.}
  \bibinfo{year}{2017}\natexlab{}.
\newblock \showarticletitle{Counting Coloured Planar Maps: Differential
  Equations}.
\newblock \bibinfo{journal}{\emph{Communications in Mathematical Physics}}
  \bibinfo{volume}{354} (\bibinfo{year}{2017}), \bibinfo{pages}{31--84}.
\newblock


\bibitem[Bernardi et~al\mbox{.}(2021)]%
        {bernardi21}
\bibfield{author}{\bibinfo{person}{Oliver Bernardi}, \bibinfo{person}{Mireille
  Bousquet-M{\'e}lou}, {and} \bibinfo{person}{Kilian Raschel}.}
  \bibinfo{year}{2021}\natexlab{}.
\newblock \showarticletitle{Counting quadrant walks via {T}utte's invariant
  method}.
\newblock \bibinfo{journal}{\emph{Combinatorial Theory}}  \bibinfo{volume}{1}
  (\bibinfo{year}{2021}), \bibinfo{pages}{\#3}.
\newblock


\bibitem[Bostan et~al\mbox{.}(2007)]%
        {bostan07}
\bibfield{author}{\bibinfo{person}{Alin Bostan},
  \bibinfo{person}{Fr{\'e}d{\'e}ric Chyzak}, \bibinfo{person}{Bruno Salvy},
  \bibinfo{person}{Gr{\'e}goire Lecerf}, {and} \bibinfo{person}{{\'E}ric
  Schost}.} \bibinfo{year}{2007}\natexlab{}.
\newblock \showarticletitle{Differential Equations for Algebraic Functions}. In
  \bibinfo{booktitle}{\emph{Proceedings of ISSAC'07}}.
  \bibinfo{publisher}{ACM}, \bibinfo{pages}{25--32}.
\newblock


\bibitem[Bostan and Jim{\'e}nez-Pastor(2020)]%
        {bostan20a}
\bibfield{author}{\bibinfo{person}{Alin Bostan} {and} \bibinfo{person}{Antonio
  Jim{\'e}nez-Pastor}.} \bibinfo{year}{2020}\natexlab{}.
\newblock \showarticletitle{On the exponential generating function of laballed
  trees}.
\newblock \bibinfo{journal}{\emph{Comptes Rendus Math{\'e}matique}}
  \bibinfo{volume}{358}, \bibinfo{number}{9--10} (\bibinfo{year}{2020}),
  \bibinfo{pages}{1005--1009}.
\newblock


\bibitem[Boulier(1996)]%
        {boulier96}
\bibfield{author}{\bibinfo{person}{Francois Boulier}.}
  \bibinfo{year}{1996}\natexlab{}.
\newblock \showarticletitle{An optimization of {S}eidenberg's elimination
  algorithm in differential algebra}.
\newblock \bibinfo{journal}{\emph{Mathematics and Computers in Simulation}}
  \bibinfo{volume}{42} (\bibinfo{year}{1996}), \bibinfo{pages}{439--448}.
\newblock


\bibitem[Bousquet-M{\'e}lou and Prince(2022)]%
        {melou20}
\bibfield{author}{\bibinfo{person}{Mireille Bousquet-M{\'e}lou} {and}
  \bibinfo{person}{Andrew~Elvey Prince}.} \bibinfo{year}{2022}\natexlab{}.
\newblock \showarticletitle{The generating function of planar {E}ulerian
  orientations}.
\newblock \bibinfo{journal}{\emph{Journal of Combinatorial Theory A}}
  \bibinfo{volume}{172} (\bibinfo{year}{2022}), \bibinfo{pages}{105183}.
\newblock


\bibitem[Carra-Ferro(1997)]%
        {carraFerro97}
\bibfield{author}{\bibinfo{person}{G. Carra-Ferro}.}
  \bibinfo{year}{1997}\natexlab{}.
\newblock \showarticletitle{A resultant theory for the systems of two ordinary
  algebraic differential equations}.
\newblock \bibinfo{journal}{\emph{Applicable Algebra in Engineering,
  Communications and Computing}} \bibinfo{volume}{8}, \bibinfo{number}{6}
  (\bibinfo{year}{1997}), \bibinfo{pages}{539--560}.
\newblock


\bibitem[Carra-Ferro(2007)]%
        {carraFerro07}
\bibfield{author}{\bibinfo{person}{G. Carra-Ferro}.}
  \bibinfo{year}{2007}\natexlab{}.
\newblock \showarticletitle{A survey on differential Gr{\"o}bner bases}. In
  \bibinfo{booktitle}{\emph{Gr{\"o}bner bases in Symbolic Analysis}}.
  \bibinfo{publisher}{De Gruyter}, \bibinfo{pages}{77--108}.
\newblock


\bibitem[Chen et~al\mbox{.}(2024)]%
        {chen24}
\bibfield{author}{\bibinfo{person}{Shaoshi Chen}, \bibinfo{person}{Hanqian
  Fang}, \bibinfo{person}{Sergey Kitaev}, {and} \bibinfo{person}{Candice~X.T.
  Zhang}.} \bibinfo{year}{2024}\natexlab{}.
\newblock \bibinfo{booktitle}{\emph{Patterns in Multi-dimensional
  Permutations}}.
\newblock \bibinfo{type}{{T}echnical {R}eport} 2411.02897.
  \bibinfo{institution}{ArXiv}.
\newblock


\bibitem[Chen et~al\mbox{.}(2013)]%
        {jaroschek13a}
\bibfield{author}{\bibinfo{person}{Shaoshi Chen}, \bibinfo{person}{Maximilian
  Jaroschek}, \bibinfo{person}{Manuel Kauers}, {and}
  \bibinfo{person}{Michael~F. Singer}.} \bibinfo{year}{2013}\natexlab{}.
\newblock \showarticletitle{Desingularization Explains Order-Degree Curves for
  {Ore} Operators}. In \bibinfo{booktitle}{\emph{Proc. ISSAC'13}}.
  \bibinfo{publisher}{ACM}, \bibinfo{pages}{157--164}.
\newblock


\bibitem[Chen and Kauers(2012)]%
        {chen12c}
\bibfield{author}{\bibinfo{person}{Shaoshi Chen} {and} \bibinfo{person}{Manuel
  Kauers}.} \bibinfo{year}{2012}\natexlab{}.
\newblock \showarticletitle{Order-Degree Curves for Hypergeometric Creative
  Telescoping}. In \bibinfo{booktitle}{\emph{Proceedings of ISSAC'12}}.
  \bibinfo{publisher}{ACM}, \bibinfo{pages}{122--129}.
\newblock


\bibitem[Cox et~al\mbox{.}(2007)]%
        {CoLiOS07}
\bibfield{author}{\bibinfo{person}{David Cox}, \bibinfo{person}{John Little},
  {and} \bibinfo{person}{Donal O’Shea}.} \bibinfo{year}{2007}\natexlab{}.
\newblock \bibinfo{booktitle}{\emph{Ideals, Varieties, and Algorithms. An
  Introduction to Computational Algebraic Geometry and Commutative Algebra}}.
\newblock \bibinfo{publisher}{Springer}.
\newblock
\urldef\tempurl%
\url{https://link.springer.com/book/10.1007/978-0-387-35651-8}
\showURL{%
\tempurl}


\bibitem[Denef and Lipshitz(1984)]%
        {denef84}
\bibfield{author}{\bibinfo{person}{J. Denef} {and} \bibinfo{person}{L.
  Lipshitz}.} \bibinfo{year}{1984}\natexlab{}.
\newblock \showarticletitle{Power Series Solutions of Algebraic Differential
  Equations}.
\newblock \bibinfo{journal}{\emph{Math. Ann.}}  \bibinfo{volume}{267}
  (\bibinfo{year}{1984}), \bibinfo{pages}{213--238}.
\newblock


\bibitem[Denef and Lipshitz(1989)]%
        {denef89}
\bibfield{author}{\bibinfo{person}{J. Denef} {and} \bibinfo{person}{L.
  Lipshitz}.} \bibinfo{year}{1989}\natexlab{}.
\newblock \showarticletitle{Decision Problems for Differential Equations}.
\newblock \bibinfo{journal}{\emph{The Journal of Symbolic Logic}}
  \bibinfo{volume}{54}, \bibinfo{number}{3} (\bibinfo{year}{1989}),
  \bibinfo{pages}{941--950}.
\newblock


\bibitem[Kauers(2014)]%
        {kauers14f}
\bibfield{author}{\bibinfo{person}{Manuel Kauers}.}
  \bibinfo{year}{2014}\natexlab{}.
\newblock \showarticletitle{Bounds for {D}-finite Closure Properties}. In
  \bibinfo{booktitle}{\emph{Proceedings of ISSAC'14}}.
  \bibinfo{publisher}{ACM}, \bibinfo{pages}{288--295}.
\newblock


\bibitem[Kauers(2023)]%
        {kauers23c}
\bibfield{author}{\bibinfo{person}{Manuel Kauers}.}
  \bibinfo{year}{2023}\natexlab{}.
\newblock \bibinfo{booktitle}{\emph{D-Finite Functions}}.
\newblock \bibinfo{publisher}{Springer}.
\newblock


\bibitem[Kolchin(1973)]%
        {kolchin73}
\bibfield{author}{\bibinfo{person}{E.~R. Kolchin}.}
  \bibinfo{year}{1973}\natexlab{}.
\newblock \bibinfo{booktitle}{\emph{Differential Algebra and Algebraic
  Groups}}.
\newblock \bibinfo{publisher}{Academic Press}.
\newblock


\bibitem[Mukhina and Pogudin(2025)]%
        {mukhina25}
\bibfield{author}{\bibinfo{person}{Yulia Mukhina} {and} \bibinfo{person}{Gleb
  Pogudin}.} \bibinfo{year}{2025}\natexlab{}.
\newblock \bibinfo{booktitle}{\emph{Projecting dynamical systems via a support
  bound}}.
\newblock \bibinfo{type}{{T}echnical {R}eport} 2501.13680.
  \bibinfo{institution}{ArXiv}.
\newblock


\bibitem[Ollivier(1991)]%
        {ollivier91}
\bibfield{author}{\bibinfo{person}{Francois Ollivier}.}
  \bibinfo{year}{1991}\natexlab{}.
\newblock \bibinfo{booktitle}{\emph{Standard bases of differential ideals}}.
\newblock \bibinfo{publisher}{Springer}.
\newblock


\bibitem[Ritt(1950)]%
        {ritt50}
\bibfield{author}{\bibinfo{person}{Josef~F. Ritt}.}
  \bibinfo{year}{1950}\natexlab{}.
\newblock \bibinfo{booktitle}{\emph{Differential Algebra}}.
\newblock \bibinfo{publisher}{American Mathematical Society, Colloquium
  Publications}.
\newblock


\bibitem[Rueda(2016)]%
        {rueda16}
\bibfield{author}{\bibinfo{person}{Sonia Rueda}.}
  \bibinfo{year}{2016}\natexlab{}.
\newblock \showarticletitle{Differential elimination by differential
  specialization of {S}ylvester style matrices}.
\newblock \bibinfo{journal}{\emph{Advances in Applied Mathematics}}
  \bibinfo{volume}{72} (\bibinfo{year}{2016}), \bibinfo{pages}{4--37}.
\newblock


\bibitem[Seidenberg(1956)]%
        {seidenberg56}
\bibfield{author}{\bibinfo{person}{A. Seidenberg}.}
  \bibinfo{year}{1956}\natexlab{}.
\newblock \showarticletitle{An elimination theory for differential algebra}.
\newblock \bibinfo{journal}{\emph{Univ. California Publ. Math}}
  \bibinfo{volume}{III}, \bibinfo{number}{22} (\bibinfo{year}{1956}),
  \bibinfo{pages}{31--38}.
\newblock


\bibitem[Shafarevich and Reid(1994)]%
        {Shafarevich94}
\bibfield{author}{\bibinfo{person}{Igor~R. Shafarevich} {and}
  \bibinfo{person}{Miles Reid}.} \bibinfo{year}{1994}\natexlab{}.
\newblock \bibinfo{booktitle}{\emph{Basic algebraic geometry 1 (2nd, revised
  and expanded ed.)}}.
\newblock \bibinfo{publisher}{Springer-Verlag}, \bibinfo{address}{Berlin,
  Heidelberg}.
\newblock
\showISBNx{0387548122}


\bibitem[van~der Hoeven(2019)]%
        {hoeven19}
\bibfield{author}{\bibinfo{person}{Joris van~der Hoeven}.}
  \bibinfo{year}{2019}\natexlab{}.
\newblock \showarticletitle{Computing with {D}-algebraic power series}.
\newblock \bibinfo{journal}{\emph{AAECC}}  \bibinfo{volume}{30}
  (\bibinfo{year}{2019}), \bibinfo{pages}{17--49}.
\newblock


\end{thebibliography}
\end{document}